\documentclass[11pt]{article}

\title{Noisy Subspace Clustering via Thresholding}

\author{Reinhard Heckel and Helmut B\"olcskei
 \\[0.5em]
  \multicolumn{1}{p{.7\textwidth}}{\centering Dept. of IT \& EE, ETH Zurich, Switzerland \\[0.5em] Email: \{heckel,boelcskei\}@nari.ee.ethz.ch}}

\date{}
\usepackage{setspace}

\renewcommand{\vspace}[1]{}

\date{}
\usepackage{setspace}

\usepackage{amssymb}

\usepackage[figuresright]{rotating}

\usepackage{amsmath,amssymb}
\usepackage{graphicx}
\usepackage{dcolumn}
\usepackage{bm}
\usepackage{amscd,amsthm}
\usepackage{ifthen}

\usepackage{mathtools} 

\usepackage[applemac]{inputenc}

\usepackage{tikz}
\usepackage{pgfplots}
\usepackage{subfigure}
\usetikzlibrary{pgfplots.groupplots}

\usepackage{bm}

\usepackage{hyperref}
\hypersetup{
    colorlinks=false,       
    linkcolor=red,          
    citecolor=green,        
    filecolor=magenta,      
    urlcolor=cyan           
}

\newcommand{\US}[1]{S^{#1-1}} 

\newcommand\PR[1]{\ensuremath{ {\mathrm{P}}\!\left[#1\right]}}

\newcommand\Tex{}
\newcommand\EX[2][\Tex]{
\ifthenelse{\equal{#1}{}}{{\mathbb E}\left[#2\right]}{\ensuremath{{\mathbb E}_{#1}\left[ #2\right]}}}

\newcommand\Var[2][\Tex]{
\ifthenelse{\equal{#1}{}}{{\mathrm{Var} }[#2]}{\ensuremath{\mathrm{Var}_{#1}\left[ #2\right]}}}

\newcommand\ignore[1]{}

\newcommand\defeq{\coloneqq}

\newcommand{\reals}{\mathbb R} 
 
\newtheorem{definition}{Definition}

\newtheorem{theorem}{Theorem}
\newtheorem*{algorithm*}{Algorithm}

\newcommand{\herm}[1]{{#1}^T} 

\renewcommand{\l}{l}
\renewcommand{\d}{d} 

\newcommand{\aff}{\mathrm{aff}}

\newcommand{\GN}[2]{ \mathcal N(#1,#2) }

 \newcommand{\mbf}[0]{\mathbf }

\newcommand\norm[2][\Tnorm]{\ensuremath{{\left\Vert #2 \right\Vert}_{#1}}}

\newcommand\Tinnerprod{}
\newcommand{\innerprod}[3][\Tinnerprod]{\ifthenelse{\equal{#1}{}}{\ensuremath{\left<#2,#3\right>}}{\ensuremath{\left<#2,#3\right>_{#1}}}}

\newcommand{\cS}{S} 
\renewcommand{\d}{d} 
\newcommand{\X}{\mathcal X} 
\renewcommand{\O}{\mathcal O} 

\usepackage{bm}

\newcommand\vect[1]{\mathbf #1}

\newcommand{\va}{\vect{a}}

\newcommand{\ve}{\vect{e}}

\newcommand{\vv}{\vect{v}}  

\newcommand{\vx}{\vect{x}}  
\newcommand{\vy}{\vect{y}}  
\newcommand{\vz}{\vect{z}}

\newcommand{\mA}{\vect{A}}

\newcommand{\mI}{\vect{I}}

\newcommand{\mU}{\vect{U}}

\renewcommand{\S}{\mathcal T}

\newcommand{\q}{q} 

\begin{document}
\maketitle

\begin{abstract}
We consider the problem of clustering noisy high-dimensional data points into a union of low-dimensional subspaces and a set of outliers. 
The number of subspaces, their dimensions, and their orientations are unknown. A probabilistic performance analysis of the thresholding-based subspace clustering (TSC) algorithm introduced recently in \cite{heckel_subspace_2012} shows 
that TSC succeeds in the noisy case, even when the subspaces intersect. 
Our results reveal an explicit tradeoff between the allowed noise level and the affinity of the subspaces. 
We furthermore find that the simple outlier detection scheme introduced in \cite{heckel_subspace_2012} provably succeeds in the noisy case. 
\end{abstract}

\section{Introduction}
\label{sec:intro}

Suppose we are given noisy observations of $N$ data points in $\reals^m$, where each data point either lies in a union of low-dimensional linear\footnote{Note that an affine subspace of dimension $d$ lies in a $(d+1)$-dimensional linear subspace. Assuming the subspaces to be linear therefore comes without loss of generality.} subspaces $\cS_l$ of $\reals^m$, $l=1,...,L$, or is an outlier. Assume that the association of the data points to the subspaces $\cS_l$ and to the set of outliers, the number of subspaces, and their orientations and dimensions are all unknown. 
We consider the problem of identifying the outliers and clustering the remaining data points, i.e., finding their assignment to the subspaces $\cS_l$. 
Once these associations have been identified, it is straightforward to extract approximations\footnote{Since we have access to noisy observations only, we cannot expect to recover the $\cS_l$ exactly. } of the subspaces $\cS_l$ through principal component analysis (PCA). 
The problem we consider is known as subspace clustering 
and has applications 
in, e.g., unsupervised learning, image processing, disease detection 
and, in particular, computer vision, e.g., motion segmentation \cite{vidal_motion_2004,rao_motion_2008} or clustering of images under varying illumination conditions \cite{ho_clustering_2003}. 
Numerous approaches to subspace clustering are available in the literature, including algebraic, statistical, and spectral clustering methods; we refer to \cite{vidal_subspace_2011} for an excellent overview. 

Spectral clustering (SC) methods (see \cite{luxburg_tutorial_2007} for an introduction) have found particularly widespread use. 
Central to SC is the construction of an adjacency matrix $\mA \in \reals^{N\times N}$, where the $(i,j)$th entry of $\mA$ measures the similarity between the data points $\vx_i$ and $\vx_j$, e.g., based on an appropriate distance measure. 
Clustering is then accomplished by identifying the connected components of the graph $G$ with adjacency matrix $\mA$. This is done through a singular value decomposition of the Laplacian of $G$ followed by $k$-means clustering \cite{luxburg_tutorial_2007}. 

We single out two recently proposed SC methods, namely the sparse subspace clustering (SSC) algorithm, introduced by Elhamifar and Vidal \cite{elhamifar_sparse_2009,elhamifar_sparse_2012}, which relies on a clever construction of $\mA$ inspired by ideas from sparse signal recovery, and an algorithm introduced by Liu et al. \cite{liu_robust_2010} that constructs  $\mA$ via a low-rank representation (LRR) of the data points. SSC provably succeeds, in the noiseless case, under very general conditions, as shown in
\cite{soltanolkotabi_geometric_2011} via an elegant (geometric function) analysis. Most importantly, the probabilistic analysis in \cite{soltanolkotabi_geometric_2011} reveals that SSC succeeds even when the subspaces $\cS_l$ intersect\footnote{The linear subspaces $\cS_l$ and $\cS_k$ are said to intersect if $\cS_l \cap \cS_k \neq \{\vect{0}\}$.}. 
A deterministic performance analysis reported in \cite{liu_robust_2010} 
shows that LRR succeeds provided the subspaces are independent, which implies that the subspaces must not intersect. 

While in the noiseless case analytical results are available for some clustering algorithms, the literature is essentially void of theoretical results on the performance of clustering algorithms in the presence of noise. Vidal noted in \cite{vidal_subspace_2011} that ``the development of theoretical sound algorithms [...] in the presence of noise and outliers is a very important open challenge''. 
A significant step towards addressing this challenge was reported recently in \cite{soltanolkotabi_robust_2013}, posted while the present manuscript was being finalized. The robust SSC (RSSC) algorithm in \cite{soltanolkotabi_robust_2013} essentially replaces the $\ell_1$-minimization steps in SSC by $\ell_1$-penalized least squares, i.e., LASSO, steps. The RSSC algorithm provably clusters data points corrupted by Gaussian noise under quite general conditions on the relative orientations of the subspaces $\cS_l$ (in particular, the $\cS_l$ are allowed to intersect) 
and on the number of points in each subspace. 
However, the construction of the adjacency matrix $\mA$ requires the solution of $N$ $\ell_1$-minimization problems in SSC and $N$ LASSO instances in RSSC; this poses significant computational challenges for large data sets. In the LRR algorithm (and its variant for the noisy case) the construction of $\mA$ requires the minimization of the nuclear norm of an $N\times N$ matrix, again resulting in significant computational challenges for large data sets. 
%
A computationally much less demanding SC-based algorithm was introduced recently in \cite{heckel_subspace_2012}. This algorithm, termed thresholding-based subspace clustering (TSC), applies SC to an adjacency matrix $\mA$ obtained by thresholding correlations between the data points, i.e., by finding the nearest neighbors (in terms of correlation) of each data point. 
For the noiseless case it was shown in \cite{heckel_subspace_2012} that TSC provably succeeds under quite general conditions, in particular, even when the subspaces intersect. 
While SSC shares these desirable properties, and has essentially identical performance guarantees, TSC is computationally much less demanding, as the construction of the adjacency matrix in TSC requires the computation of $N^2$ inner products followed by thresholding only.  
\paragraph*{Contributions:}
The aim of this paper is to analyze the performance of TSC in the noisy case. Specifically, we show that TSC provably succeeds under the influence of additive Gaussian noise, and does so under very general conditions on the relative orientations of the subspaces and on the number of points in the subspaces. In particular, the subspaces are allowed to intersect. Our analysis furthermore shows that the more distinct the orientations of the subspaces, the more noise TSC tolerates. Interestingly, TSC can succeed even under massive noise on the data points, 
provided that the subspaces are sufficiently low-dimensional. 
Finally, we show that the simple scheme for outlier detection introduced in \cite{heckel_subspace_2012} provably succeeds  in the noisy case as well.  
Detailed proofs of the theorems in this paper, additional results on clustering noisy data points, and numerical results for real data sets are provided in \cite{heckel_robust_2013}.

\paragraph*{Notation:} We use lowercase boldface letters to denote (column) vectors, e.g., $\vx$, and uppercase boldface letters to designate matrices, e.g., $\mA$. 
For the vector $\vx$, $[\vx]_q$ and $x_q$ denote its $q$th entry and for the matrix $\mA$, $\mA_{ij}$ stands for the entry in the $i$th row and $j$th column. 
The spectral norm of $\mA$ is $\norm[2\to 2]{\mA} \defeq\;$ $\max_{\norm[2]{\vv} = 1  } \norm[2]{\mA \vv}$, its Frobenius norm is $\norm[F]{\mA} \defeq (\sum_{i,j} |\mA_{ij}|^2 )^{1/2}$, and $\mI$ denotes the identity matrix. 
The superscript $\herm{}$ stands for transposition. $\log(\cdot)$ refers to the natural logarithm, and $x \land y$ denotes the minimum of $x$ and $y$. 
 The cardinality of the set $\S$ is $|\S|$. 
The set $\{1,...,N\}$ is written as $[N]$. 
We use $\mathcal N( \boldsymbol{\mu},\boldsymbol{\Sigma})$ to designate a Gaussian random vector with mean $\boldsymbol{\mu}$ and covariance matrix $\boldsymbol{\Sigma}$. 
The unit sphere in $\reals^m$ is $\US{m} \defeq \{ \vx \in \reals^m \colon \norm[2]{\vx} = 1 \}$. 

\vspace{-0.1cm}
\section{Problem statement and the TSC algorithm}

The formal statement of the problem we consider is as follows. 
Suppose we are given a set of $N$ data points in $\reals^m$, denoted by $\X$, and assume that 
$
\X = \X_1 \cup ...  \cup  \X_L \cup  \O
$, 
where $\O$ denotes a set of outliers and the points in $\X_l, l \in [L]$, are given by
$
\vx_j^{(l)} = \vy_j^{(l)}  + \,\ve^{(l)}_j 
$, $j\in [n_l]$, where $n_\l = |\X_l|$. Here, $\vy_j^{(l)} \in \cS_l$ with $\cS_l$ a $d_l$-dimensional subspace of $\reals^m$, and $\ve^{(l)}_j \in \reals^m$ is i.i.d.~(across $l$ and $j$)~$\GN{\mbf 0}{(\sigma^2/m)\mI}$ noise. 
The association of the points in $\X$ with the $\X_l$ and $\O$, the number of points in each subspace $n_l = |\X_l|$, the number of subspaces $L$, their dimensions $d_l$, and their orientations are all unknown. We want to cluster the (noisy) points in $\X$, i.e., find their assignment to the sets $\X_l, \O$.

We next briefly summarize the TSC algorithm introduced in \cite{heckel_subspace_2012}. 
The formulation of the TSC algorithm given below assumes that outliers have already been removed from $\X$, e.g., through the outlier detection scheme 
discussed in Sec.~\ref{sec:detoutl}. Moreover, for Step 1 below to make sense, we assume that the data points in $\X$ are either normalized or of comparable norm. This assumption is not restrictive as the data points can be normalized prior to clustering. 

\vspace{0.1cm}
{\bf The TSC algorithm.} 
Given a set of data points $\X$ and the parameter $q$ (the choice of $q$ is discussed below), perform the following steps:

{\bf Step 1:} For every $\vx_j \in \X$, find the set $\S_j \subset [N] \setminus j$ of cardinality $q$ defined through
\begin{equation*}
\left| \innerprod{\vx_j}{ \vx_i} \right| \geq \left| \innerprod{\vx_j}{ \vx_p} \right| \text{ for all }  i \in \S_j \text{ and all } p \notin \S_j
\end{equation*}
and let $\vz_j \in \reals^N$ be the vector with $i$th entry $\left| \innerprod{\vx_j}{ \vx_i} \right|$ if $i\in \S_j$, and $0$ if $i\notin \S_j$. 
Construct the adjacency matrix $\mA$ according to $\mA_{ij} = |  [\vz_j]_i | + | [\vz_i]_j |$. 

{\bf Step 2:} Estimate the number of subspaces using the eigengap heuristic \cite{luxburg_tutorial_2007}  according to
$
\hat L = \arg  \max_ {i=1,..., N-1}  \allowbreak(\lambda_{i+1} - \lambda_{i}),
$
where $\lambda_1\leq \lambda_2 \leq ... \leq \lambda_N$ are the eigenvalues of the normalized Laplacian of the graph with adjacency matrix $\mA$. 

{\bf Step 3:} Apply normalized SC \cite{luxburg_tutorial_2007} to $(\mA, \hat L)$. 

\vspace{0.5em}
TSC is said to succeed if the following subspace detection property holds. 
\begin{definition}
The subspace detection property holds for $\X = \X_1 \cup \,...  \,\cup \, \X_L$ and adjacency matrix $\mA$ if

\hspace{-0.1cm}{\it i.} $\mA_{ij} \!\neq \!0$ only if $\vx_i$ and $\vx_j$ belong to the same set $\X_l$

\noindent and if

\hspace{-0.1cm}{\it ii.} for every $i \in[N]$, $\mA_{ij} \neq 0$ for at least $\q$ points $\vx_j$ that belong to the same set $\X_l$ as $\vx_i$. 
\label{def:lsdp}
\end{definition}
\vspace{-0.2cm}
The subspace detection property is similar to the $\ell_1$ subspace detection property introduced in \cite{soltanolkotabi_geometric_2011}. The corresponding notion in \cite{soltanolkotabi_robust_2013} is that of $\mA$ having ``no false discoveries'' and at least $\q$ ``true discoveries'' in each row/column. 
The subspace detection property guarantees that  each node in the Graph $G$ (with adjacency matrix $\mA$) is connected to at least $q$ other nodes, all of which correspond to points within the same subspace. 
Note that even when the subspace detection property does not hold strictly, but the $\mA_{ij}$ for pairs $\vx_i,\vx_j$ belonging to different subspaces are ``small enough'', TSC may still cluster the data correctly, owing to the SC step.

\paragraph*{Choice of $q$:}
Recall that $q$ is an input parameter to the TSC algorithm. 
Choosing $q$ too small/large will lead to over/under-estimation of the number of subspaces $L$. 
Our analytical performance results ensure that the subspace detection property holds given that $\q$ is sufficiently small relative to the $n_\l$. Once this condition is satisfied, the specific choice of $\q$ does not matter in terms of our analytical performance guarantees. 

\vspace{-0.14cm}
\section{\label{sec:clustnoise}Performance guarantees}
\vspace{-0.05cm}
In order to elicit the impact of the relative orientations of the subspaces $\cS_l$ on the performance of TSC, we take the subspaces $\cS_l$ to be deterministic and choose the points in the subspaces randomly. To this end, we represent the data points in $\cS_l$ by  $\vy_j^{(l)} = \mU^{(l)} \va^{(l)}_j$, where $\mU^{(l)} \in \reals^{m\times \d_\l}$ is a deterministic orthonormal basis for the $\d_\l$-dimensional subspace $\cS_l$ and the $\va^{(l)}_j \in \reals^{\d_\l}$ are random vectors i.i.d.~uniformly distributed on $\US{\d_\l}$. Since the $\mU^{(l)}$ are orthonormal, the data points $\vy_j^{(l)}$ are  uniformly distributed on the set of points in $\cS_l$ with unit norm. 
The performance guarantees we obtain are expressed 
 in terms of the affinity between subspaces defined as
\cite[Def.~2.6]{soltanolkotabi_geometric_2011}, \cite[Def.~1.2]{soltanolkotabi_robust_2013}:
\[
\aff(\cS_k,\cS_l) \defeq \frac{1}{\sqrt{\d_k \land \d_\l }}\norm[F]{ \herm{\mU^{(k)}} \mU^{(l)}  }. 
\]
Note that $0 \leq \aff(\cS_k,\cS_l) \leq 1$, with $\aff(\cS_k,\cS_l) = 1$ if $\cS_k=\cS_l$ and $\aff(\cS_k,\cS_l) = 0$ if $\cS_k$ and $\cS_l$ are orthogonal. 
Moreover, $\aff(\cS_k,\cS_l) =\allowbreak \sqrt{ \cos^2( \theta_1) + ...+ \cos^2(\theta_{\d_k \land \d_\l})}/\sqrt{\d_k \land \d_\l}$, where $\theta_1 \leq ... \leq \theta_{\d_k \land \d_\l}$ denote the principal angles between $\cS_k$ and $\cS_l$. 
If $\cS_k$ and $\cS_l$ intersect in $p$ dimensions, i.e., if $\cS_k \cap \cS_l$ is $p$-dimensional, then $\cos(\theta_1)=...=\cos(\theta_{p})=1$ and hence $\aff(\cS_k,\cS_l) \geq \sqrt{p/(\d_k \land \d_\l)}$.

\begin{theorem}
Suppose $\X$ is obtained by choosing, for each $l\in [L]$, $n_l=\rho_l \q$, with $\rho_l\geq 6$, points corresponding to $\cS_\l$ at random according to $\vx_j^{(\l)} = \mU^{(\l)} \va^{(\l)}_j  + \ve^{(\l)}_j$, where the $\va^{(l)}_j$ are i.i.d.~uniform on $\US{\d_\l}$ and the $\ve^{(\l)}_j$ are i.i.d.~$\mathcal N( \vect{0}, (\sigma^2/m)  \mI)$ and independent of the $\va^{(\l)}_j$.
Suppose further that 
\begin{align}
\max_{k,l\colon k \neq l}  \aff(\cS_k,\cS_l)  \!+\!  \frac{\sigma(1+\sigma)}{\sqrt{\log N}} \frac{\sqrt{d_{\max}}}{\sqrt{m}}  \leq \frac{1}{12  \log N}
\label{eq:condthmnoisycase}
\end{align}
with $m \geq 6 \log N$, where $\d_{\max}  = \max_l \d_l$. 
Then, the subspace detection property holds (for $\X$) with probability at least 
$
1 - \frac{10}{N} - \sum_{l\in [L]} n_l e^{-c(n_l-1)}
$, 
where $c>0$ is an absolute constant. 
\label{thm:noisycase}
\end{theorem}
\begin{proof}
A sketch is provided in the appendix. 
\end{proof}

We next interpret Thm.~\ref{thm:noisycase} separately in the noiseless and in the noisy cases. 

\paragraph*{The noiseless case:}
In the noiseless case, i.e., for $\sigma = 0$, Thm.~\ref{thm:noisycase} states that TSC succeeds with high probability if the maximum affinity between different subspaces 
is sufficiently small, and if $\X$ contains sufficiently many points of each subspace. Note that Thm.~\ref{thm:noisycase} (for $\sigma=0$) does not impose any restrictions on the dimensions of the subspaces; the only dependence on the subspaces is via the affinity in \eqref{eq:condthmnoisycase}.  
We furthermore observe that for increasing $n_\l$, the probability of success in Thm.~\ref{thm:noisycase} increases, while Cond.~\eqref{eq:condthmnoisycase} becomes slightly harder to satisfy as the RHS of \eqref{eq:condthmnoisycase} decreases, albeit slowly, in $N = \sum_l n_l$. 
\paragraph*{The noisy case:}
In the noisy case, Thm.~\ref{thm:noisycase} states that TSC succeeds with high probability 
if the noise variance and the affinities between the subspaces are sufficiently small, and if $\X$ contains sufficiently many points of each subspace. Cond.~\eqref{eq:condthmnoisycase} nicely reflects the intuition that the more distinct the orientations of the subspaces, the more noise TSC tolerates. 
What is more, Cond.~\eqref{eq:condthmnoisycase} reveals that TSC can succeed even when the noise variance $\sigma^2$ is large, provided that $\sqrt{\d_{\max}/m}$ is sufficiently small. 

The intuition behind the factor $\sigma (1+\sigma) \sqrt{\d_{\max}/m}$ in \eqref{eq:condthmnoisycase} is as follows. 
Assume, for simplicity, that $\d_l = \d$, for all $l$, and consider the most favorable situation of orthogonal subspaces, i.e., $\aff(\cS_k,\cS_l)= 0$, for all $k\neq l$. TSC relies on the inner products between points within a given subspace to typically be larger than the inner products between points in distinct subspaces. 
First, note that $\innerprod{\vx_j}{\vx_i} = \innerprod{\vy_j}{\vy_i} + \innerprod{\ve_j}{\ve_i} + \innerprod{\vy_j}{\ve_i} + \innerprod{\ve_j}{\vy_i}$. Then, under the probabilistic data model of Thm.~\ref{thm:noisycase}, we have 
$
\left(\EX{ \left| \innerprod{\vy_j}{\vy_i} \right|^2}\right)^{1/2} = \frac{1}{\sqrt{d}}
$
if $\vy_j, \vy_i \in \cS_l$ and 
$
\innerprod{\vy_j}{\vy_i}  = 0
$
if $\vy_j \in \cS_k$ and $\vy_i \in \cS_l$, with $k\neq l$. When the terms $\innerprod{\ve_j}{\ve_i}$, $\innerprod{\vy_j}{\ve_i}$, and  $\innerprod{\ve_j}{\vy_i}$ are small relative to $\frac{1}{\sqrt{d}}$, we have a margin on the order of $\frac{1}{\sqrt{d}}$ to separate points within a given cluster from points in other clusters. Indeed, if $\frac{\sigma}{\sqrt{m}}$ is small relative to $\frac{1}{\sqrt{d}}$, $\innerprod{\vy_j}{\ve_i}$ and $\innerprod{\ve_j}{\vy_i}$ are (sufficiently) small, while $\frac{\sigma^2}{\sqrt{m}}$ being small relative to  $\frac{1}{\sqrt{d}}$ ensures that $\innerprod{\ve_j}{\ve_i}$ is (sufficiently) small. 
These two conditions are satisfied when $\sigma (1+\sigma) \sqrt{\d/m}$ is (sufficiently) small.

\paragraph*{Comparison of TSC with SSC and RSSC:}
For SSC in the noiseless and RSSC in the noisy case, results analogous to Thm.~\ref{thm:noisycase} were reported in \cite[Thm.~2.8]{soltanolkotabi_geometric_2011} and \cite[Thm.~3.1]{soltanolkotabi_robust_2013}, respectively. 
While TSC is based on a ``local'' criterion, namely the comparison of inner products of pairs of data points, SSC and RSSC employ a more ``global'' criterion by finding a sparse representation of each data point in terms of all the other data points. In the light of this observation it is interesting to see that the clustering conditions and the performance guarantees of TSC on the one hand and SSC and RSSC on the other hand, are essentially identical. 
Specifically, the clustering condition for RSSC in \cite{soltanolkotabi_robust_2013} is identical (up to constants) to \eqref{eq:condthmnoisycase} with $\sigma(1+\sigma)$ in \eqref{eq:condthmnoisycase} replaced by $\sigma$. 
We note, however, that \cite{soltanolkotabi_robust_2013} requires $\sigma$ to be bounded, an assumption not needed here. Obviously for $\sigma$ bounded, the factor $\sigma(1+\sigma)$ in Cond.~\eqref{eq:condthmnoisycase} can be replaced by $\sigma$ times a constant. 
Note that both TSC and RSSC have input parameters, $\q$ for TSC and the LASSO-weight $\lambda$ for RSSC. 
Finally, thanks to the simplicity of TSC, the proof of Thm.~\ref{thm:noisycase} is conceptually and technically less involved than the proof of the corresponding (main) result for RSSC in \cite{soltanolkotabi_robust_2013}. 

\newcommand{\p}[0]{p}
\vspace{-0.2cm}
\section{\label{sec:detoutl}Outlier detection}
\vspace{-0.1cm}
Outliers are data points that do not lie in one of the (low-dimensional) subspaces $\cS_l$ and do not exhibit low-dimensional structure. 
Here, this is accounted for by assuming random outliers distributed as $\mathcal N(\mbf 0, (1/m) \mI )$. Note that this implies that the direction of the outliers, i.e., $\vx_i/\norm[2]{\vx_i}$, is uniformly distributed on $\US{m}$, and for $m$ large, the norm of the outliers $\vx_i$ concentrates around one. 
In this section, we normalize the inliers to ensure that outlier separation can not trivially be accomplished by exploiting differences in the norm between inliers and outliers. 
We consider the outlier detection scheme introduced for the noiseless case in \cite{heckel_subspace_2012}, which declares $\vx_j$ an outlier if 
\vspace{-0.1cm}
\begin{align}
\max_{p\neq j} \left| \innerprod{\vx_p}{\vx_j}\right| <  c \sqrt{ \log N} / \sqrt{m}
\label{eq:outldetrulenc}
\vspace{-0.1cm}
\end{align} 
where $c$ is a suitably chosen constant.

\begin{theorem}
Suppose $\X$ consists of $N_0$ outliers chosen at random i.i.d.~$\mathcal N(\vect{0},(1/m) \mI)$, and of $\sum_l n_l$ inliers obtained as follows. For each $l \in [L]$, choose $n_l$ points corresponding to $\cS_\l$ at random according to $\vx_j^{(l)} = \frac{1}{\sqrt{1+\sigma^2}} \allowbreak \left(  \mU^{(l)} \va^{(l)}_j  + \ve^{(l)}_j \right)$, where the $\va^{(k)}_j$ are i.i.d.~uniform on $\US{\d_\l}$ and the $\ve^{(l)}_j$ are i.i.d.~$\mathcal N(\vect{0},(\sigma^2/m) \mI )$ distributed, and independent of the $\va^{(l)}_j$. 
Declare $\vx_j \in \X$ to be an outlier if \eqref{eq:outldetrulenc} holds with $c=2.3 \sqrt{6}$. 
Then, every outlier is detected with probability at least $1 -3 \frac{N_0}{N^2}$, where $N=N_0 + \sum_\l n_\l$. Moreover, provided that
\begin{align}
\frac{d_{\max}}{m} \leq  \frac{c_1}{(1+\sigma^2)^2 \log N}
\label{eq:condnoisyoutldet}
\end{align}
where $c_1$ is an absolute constant and $d_{\max} = \max_l \d_\l$, for each $l\in [L]$, with probability at least
$
1-  n_l\left(e^{-\frac{1}{2}\log\left(\frac{\pi}{2}\right) (n_l-1) }  +  n_l \frac{7}{N^3}\right) 
$
no inlier belonging to $\cS_l$ is misclassified as an outlier.
\label{thm:outldetnoisyc}
\end{theorem}
Due to space constraints, the proof of Thm.~\ref{thm:outldetnoisyc} is relegated to \cite{heckel_robust_2013}. 
Since \eqref{eq:condnoisyoutldet} can be rewritten as
$
N_0 \leq e^{\frac{m}{d_{\max}} \frac{c_1}{(1+\sigma^2)^2}   } - \sum_\l n_\l,
$
we can conclude that outlier detection succeeds even if the number of outliers scales exponentially in $m/\d_{\max}$, i.e., if $d_{\max}$ and $\sigma^2$ are kept constant, exponentially in the ambient dimension. For the noiseless case this was found previously in \cite{heckel_subspace_2012}. Note that this result does not need any assumptions on the relative orientations of the subspaces $\cS_l$. 
We finally remark that outlier detection can succeed even when $\sigma^2$ is large, provided that $\d_{\max}/m$ is sufficiently small. 
An outlier detection scheme for RSSC does not seem to be available. 

\section{Numerical results}

We measure performance in terms of the clustering error (CE), defined as the ratio of the number of misclassified points and the total number of points in $\X$. 
We generate $L=15$ subspaces of $\reals^{50}$ of equal dimension $d$ by choosing the corresponding orthonormal bases $\mU^{(l)} \in \reals^{m\times d}$  uniformly at random from the set of orthonormal matrices in $\reals^{m\times d}$. We set $\q=d$ and vary the number of points in each subspace $n=d \rho$ by varying $\rho$.  
The points in the individual subspaces are chosen at random according to the probabilistic model in Thm.~\ref{thm:noisycase}. 
The numerical results in Fig.~\ref{fig:noise} show that even when $\sigma^2$ is large, TSC succeeds, provided that $n$ is sufficiently large. 

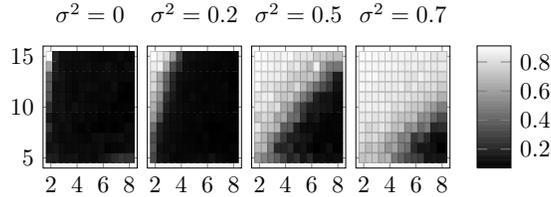
\begin{figure}
\centering
\begin{tikzpicture}[scale=0.9] 
\begin{groupplot}[group style={group size=5 by 1,horizontal sep=0.15cm,vertical sep=0.85cm,xlabels at=edge bottom, ylabels at=edge left,xticklabels at=edge bottom,yticklabels at=edge left},
width=0.18\textwidth,height=0.205\textwidth,/tikz/font=\small, colormap/blackwhite],point meta min = 0, point meta max=0.8]

  \nextgroupplot[title={$\sigma^2=0$}]
	 \addplot[mark=square*,only marks, scatter, scatter src=explicit,
	 mark size=2,colormap/blackwhite]
	 file {./CE_N0.dat};
	 	  
  
  \nextgroupplot[title={$\sigma^2=0.2$}]
	 \addplot[mark=square*,only marks, scatter, scatter src=explicit,
	 mark size=2,colormap/blackwhite]
	 file {./CE_N02.dat};

%
	 
  \nextgroupplot[title={$\sigma^2=0.5$}]
	 \addplot[mark=square*,only marks, scatter, scatter src=explicit,
	 mark size=2,colormap/blackwhite]
	 file {./CE_N05.dat};
%
	 
  \nextgroupplot[title={$\sigma^2=0.7$},colorbar]
	 \addplot[mark=square*,only marks, scatter, scatter src=explicit,
	 mark size=2,colormap/blackwhite]
	 file {./CE_N07.dat};

  \end{groupplot}
\end{tikzpicture}
\vspace{-0.3cm}
\caption{\label{fig:noise}CE as a function of the dimension of the subspaces, $d$, on the vertical and $\rho$ on the horizontal axis for different noise variances $\sigma^2$.}
\vspace{-0.3cm}
\end{figure}






\section*{\label{app:thm:noisycase}Appendix: Proof Sketch of Thm.~\ref{thm:noisycase}}
The subspace detection property is certainly satisfied if for each $\vx_i^{(l)} \in \X_l$, and for each $\X_l$, the points in the set $\S_i$ that corresponds to $\vx_i^{(l)}$ are all in $\X_l$ and $|\S_i|=\q$. The latter condition is satisfied by construction (of the set $\S_i$). Regarding the former condition, consider w.l.o.g.~$\vx_i^{(l)}$ with corresponding set $\S_i$ and define $
z_j^{(k)} \defeq \left| \innerprod{\vx_j^{(k)}}{ \vx_i^{(l)} } \right|.
$
Note that for simplicity of exposition, the notation $z_j^{(k)}$ does not reflect the dependence on $\vx_i^{(l)}$. By definition $\S_i$  corresponds to points in $\X_l$ only if 
\vspace{-0.1cm}
\begin{align}
z_{(n_\l - \q)}^{(l)} > \max_{k\neq l, j} z_{j}^{(k)} 
\label{eq:tscsdpfox2}
\end{align}
\vspace{-0.35cm}

\noindent where the order statistics $z_{(1)}^{(\l)} \leq z_{(2)}^{(\l)} \leq ...\leq z_{(n_\l-1)}^{(\l)}$ are defined by sorting the $\{z_{j}^{(\l)}\}_{j \in [n_\l] \setminus i}$ in ascending order. 
Next, we upper-bound the probability of \eqref{eq:tscsdpfox2} being violated (for a given $\vx_i^{(l)}$). A union bound over all $N$ vectors $\vx_i^{(l)}, i\in [n_\l],\; l\in [L]$, will then yield the final result. 
We start by defining  
\[
\tilde z_j^{(k)} \defeq \innerprod{  \va_j^{(k)}}{  \herm{\mU^{(k)}} \mU^{(l)} \va_i^{(l)}}
\]
and noting that
$
z_j^{(k)} = \left|  \tilde z_j^{(k)}  +  e_j^{(k)}  \right|
$
with 
\[
e_j^{(k)} = \innerprod{ \ve_j^{(k)}}{ \ve_i^{(l)} }  +\innerprod{  \ve_j^{(k)}}{   \mU^{(l)} \va_i^{(l)}}  +\innerprod{   \mU^{(k)} \va_j^{(k)} }{ \ve_i^{(l)} }.
\]
With this notation,  
\begin{align}
\PR{z_{(n_\l-\q)}^{(l)} \leq \max_{k\neq l, j} z_{j}^{(k)} }  
&\leq \!
\PR{|\tilde z_{(n_\l-\q)}^{(l)}| \!-\! \max_{j\neq i} |e_j^{(l)}|   \leq \max_{k\neq l, j} |\tilde z_{j}^{(k)}|  +  \max_{k\neq l, j}  |e_{j}^{(k)}|   } \nonumber \\
&\leq 
\PR{|\tilde z_{(n_\l-\q)}^{(l)}| \leq \frac{\nu}{\sqrt{\d_\l}}   } \label{eq:assumeasdage} \\
&+\! \PR{ \alpha + 2 \epsilon \leq \max_{j\neq i} |e_j^{(l)}|   \!+\! \max_{k\neq l, j} |\tilde z_{j}^{(k)}|  \!+\!  \max_{k\neq l, j}  |e_{j}^{(k)}|   } \nonumber \\
&\leq 
\PR{|\tilde z_{(n_\l-\q)}^{(l)}| \leq \frac{\nu}{\sqrt{\d_\l}}   }  +\PR{ \max_{k\neq l, j} |\tilde z_{j}^{(k)}|  \geq \alpha  }  \nonumber \\
&+ \underbrace{\PR{\max_{j\neq i} |e_{j}^{(l)}| \geq  \epsilon   }   +  \PR{ \max_{k\neq l, j} |e_{j}^{(k)}| \geq  \epsilon  }}_{\leq \sum_{(j,k)\neq (i,l)} \PR{ \left|  e_j^{(k)} \right|   \geq  \epsilon }  } \label{eq:prtoboundt}
\end{align}
where $\alpha, \epsilon,$ and $\nu$ are chosen later and for \eqref{eq:assumeasdage} 
we used that for two random variables $X,Y$ and constants $\phi,\varphi$ satisfying $\phi \geq \varphi$ we have that 
\[
\PR{X \leq Y} 
\leq \PR{ \{X \leq \phi \} \cup \{\varphi \leq Y\} } 
\leq \PR{X\leq \phi} + \PR{\varphi \leq Y}.
\]
Specifically, for \eqref{eq:assumeasdage} we set $\phi=\frac{\nu}{\sqrt{\d_\l}}$ and $\varphi = \alpha + 2 \epsilon$, which leads to the assumption  $\alpha+ 2 \epsilon \leq \frac{\nu}{\sqrt{\d_\l}}$, resolved below. 
In the next steps of the proof, detailed in \cite{heckel_robust_2013}, we upper-bound the individual terms in \eqref{eq:prtoboundt}.
\paragraph*{Step 1:} With $\epsilon = \frac{2 \sigma(1+\sigma)}{\sqrt{m}} \beta$, for $\beta \geq \frac{1}{\sqrt{2\pi}}$ satisfying $\beta / \sqrt{m} \leq 1$, we have 
\begin{align}
\PR{ \left|  e_j^{(k)} \right|   \geq  \epsilon } \leq 7 e^{-\frac{\beta^2}{2}}. 
\label{eq:boundonnoise}
\end{align}
\paragraph*{Step 2:} Setting \\
\[
\alpha = \frac{\beta(1+\beta)}{\sqrt{d_l}}  \max_{k \neq l} \allowbreak \frac{1}{\sqrt{d_k}}  \norm[F]{ \herm{\mU^{(k)}} \mU^{(l)} }
\] 
for $\beta \geq 0$, we get
\begin{align}
&\PR{\max_{k\neq l, j} |\tilde z_{j}^{(k)}|  \geq \alpha }
\leq 3N e^{- \frac{\beta^2}{2}} .  \label{eq:maxboundadfa}
\end{align}
\paragraph*{Step 3:}
For $\nu = 2/3$ and $n_l/\q = \rho_l \geq 6$ there is a constant $c(\rho_l,\nu) > 1/20$ such that
\begin{align} 
\PR{|\tilde z_{(n_l-\q)}^{(l)}| \leq \frac{\nu}{\sqrt{d_l}} } \leq e^{-c (n_l-1)} \label{eq:znidlexa}.
\end{align}
Using \eqref{eq:boundonnoise}, \eqref{eq:maxboundadfa}, and \eqref{eq:znidlexa} in \eqref{eq:prtoboundt} and setting $\beta = \sqrt{6 \log N }$  yields 
\begin{align}
\PR{z_{(n_\l - \q)}^{(l)} \leq \max_{k\neq l, j} z_{j}^{(k)} } 
\leq \frac{10}{N^2} + e^{-c(n_\l-1)} .
\end{align}
Taking the union bound over all vectors $\vx_i^{(l)}, i \in [n_\l], \l \in [L]$, yields the desired lower bound on the probability of the subspace detection property to hold. 

Recall that for \eqref{eq:assumeasdage}  we imposed the condition  $\alpha+ 2 \epsilon \leq \frac{\nu}{\sqrt{\d_l}}$. With our choices for $\epsilon, \alpha$, and $\nu$ this condition 
is implied by \eqref{eq:condthmnoisycase}, for all $\l \in [L]$. This concludes the proof.

\end{document}